%% file: main-extended.tex
\documentclass[dratfcls,onecolumn]{IEEEtran}
\input{preamble}

\newcommand{\xv}{x}
\newcommand{\yv}{y}
\newcommand{\xvhat}{\hat{x}}
\newcommand{\Av}{A}
\newcommand{\Xc}{\mathcal{X}}
\newcommand{\RR}{\mathbb{R}}
\newcommand{\Shat}{\hat{S}} 

\usepackage{balance}
\usepackage{algorithm}
\usepackage{algorithmic}
\usepackage{cleveref}

\usepackage{amsmath,amssymb,amsthm,dsfont,mathrsfs}
\usepackage{makecell}
\usepackage{array,multirow,graphicx}

\newtheorem{theorem}{Theorem}
\newtheorem{lemma}{Lemma}
\newtheorem{corollary}{Corollary}

\newtheorem{definition}{Definition}

\theoremstyle{definition}
\newtheorem{example}{Example}
\theoremstyle{remark}

\newtheorem*{note*}{Note}

\newtheorem*{remark*}{Remark}

\pgfmathdeclarefunction{gauss}{2}{\pgfmathparse{1/(#2*sqrt(2*pi))*exp(-((x-#1)^2)/(2*#2^2))}%
}
\pgfmathdeclarefunction{rect}{1}{\pgfmathparse{#1}%
}
\begin{document}

\title{Universal 1-Bit Compressive Sensing \\ for Bounded Dynamic Range Signals}  
 
\author{Sidhant Bansal$^1$, Arnab Bhattacharyya$^1$, Anamay Chaturvedi$^2$, and Jonathan Scarlett$^1$, \\
    $^1$National University of Singapore, $^2$Northeastern University}

\maketitle 

\begin{abstract}
    A {\em universal 1-bit compressive sensing (CS)} scheme consists of a measurement matrix $A$ such that all signals $x$ belonging to a particular class can be approximately recovered from $\textrm{sign}(Ax)$. 1-bit CS models extreme quantization effects where only one bit of information is revealed per measurement. 
     
    We focus on universal support recovery for 1-bit CS in the case of {\em sparse} signals with bounded {\em dynamic range}. Specifically, a vector $x \in \bbR^n$ is said to have sparsity $k$ if it has at most $k$ nonzero entries, and dynamic range $R$ if the ratio between its largest and smallest nonzero entries is at most $R$ in magnitude. Our main result shows that if the entries of the measurement matrix $A$ are i.i.d.~Gaussians, then under mild assumptions on the scaling of $k$ and $R$, the number of measurements needs to be $\tilde{\Omega}(Rk^{3/2})$ to recover the support of $k$-sparse signals with dynamic range $R$ using $1$-bit CS.  In addition, we show that a near-matching $O(R k^{3/2} \log n)$ upper bound follows as a simple corollary of known results.  The $k^{3/2}$ scaling contrasts with the known lower bound of $\tilde{\Omega}(k^2 \log n)$ for the number of measurements to recover the support of arbitrary $k$-sparse signals.  
\end{abstract}

\input{intro.tex}

\input{prelims.tex}

\input{gaussian.tex}

\input{balanced.tex}

\input{rademacher.tex}

\medskip
{\bf Acknowledgment.} J.~Scarlett was funded by an NUS Early Career Research Award (grant R-252-000-A36-133).  A.~Bhattacharyya was supported in part by an MOE Tier II award and an Amazon faculty research award.


\balance
\bibliographystyle{IEEEtran}
\bibliography{biblio,JS_References}

\end{document}

%% file: preamble.tex
\newcommand{\ignore}[1]{}

\usepackage{epsfig,epsf,psfrag}
\usepackage{amsmath,graphicx,bm,xcolor,url}
\usepackage{array}
\usepackage{verbatim}
\usepackage{bm}
\usepackage{algorithmic, cite,bbm}
\usepackage{algorithm}
\usepackage{verbatim}
\usepackage{mathrsfs}
\usepackage{url}
\usepackage{hyperref}
\usepackage{epstopdf}

\usepackage[capitalise]{cleveref}
\usepackage{comment}
\usepackage{pgfplots}

\usepackage{todonotes}

\newcounter{mynotes}
\catcode`~=11 \def\UrlSpecials{\do\~{\kern -.15em\lower .7ex\hbox{~}\kern .04em}} \catcode`~=13 

\allowdisplaybreaks[4]

\newcommand{\bbR}{\mathbb{R}}

\DeclareMathAlphabet{\mathbsf}{OT1}{cmss}{bx}{n}
\DeclareMathAlphabet{\mathssf}{OT1}{cmss}{m}{sl}

\DeclareSymbolFont{bsfletters}{OT1}{cmss}{bx}{n}  
\DeclareSymbolFont{ssfletters}{OT1}{cmss}{m}{n}
\DeclareMathSymbol{\bsfGamma}{0}{bsfletters}{'000}
\DeclareMathSymbol{\ssfGamma}{0}{ssfletters}{'000}
\DeclareMathSymbol{\bsfDelta}{0}{bsfletters}{'001}
\DeclareMathSymbol{\ssfDelta}{0}{ssfletters}{'001}
\DeclareMathSymbol{\bsfTheta}{0}{bsfletters}{'002}
\DeclareMathSymbol{\ssfTheta}{0}{ssfletters}{'002}
\DeclareMathSymbol{\bsfLambda}{0}{bsfletters}{'003}
\DeclareMathSymbol{\ssfLambda}{0}{ssfletters}{'003}
\DeclareMathSymbol{\bsfXi}{0}{bsfletters}{'004}
\DeclareMathSymbol{\ssfXi}{0}{ssfletters}{'004}
\DeclareMathSymbol{\bsfPi}{0}{bsfletters}{'005}
\DeclareMathSymbol{\ssfPi}{0}{ssfletters}{'005}
\DeclareMathSymbol{\bsfSigma}{0}{bsfletters}{'006}
\DeclareMathSymbol{\ssfSigma}{0}{ssfletters}{'006}
\DeclareMathSymbol{\bsfUpsilon}{0}{bsfletters}{'007}
\DeclareMathSymbol{\ssfUpsilon}{0}{ssfletters}{'007}
\DeclareMathSymbol{\bsfPhi}{0}{bsfletters}{'010}
\DeclareMathSymbol{\ssfPhi}{0}{ssfletters}{'010}
\DeclareMathSymbol{\bsfPsi}{0}{bsfletters}{'011}
\DeclareMathSymbol{\ssfPsi}{0}{ssfletters}{'011}
\DeclareMathSymbol{\bsfOmega}{0}{bsfletters}{'012}
\DeclareMathSymbol{\ssfOmega}{0}{ssfletters}{'012}

\newcommand{\ceil}[1]{\lceil{#1}\rceil}
\newcommand{\floor}[1]{\lfloor{#1}\rfloor}

\DeclareMathOperator{\supp}{supp}

\newcommand{\qednew}{\nobreak \ifvmode \relax \else
      \ifdim\lastskip<1.5em \hskip-\lastskip
      \hskip1.5em plus0em minus0.5em \fi \nobreak
      \vrule height0.75em width0.5em depth0.25em\fi}

%% file: intro.tex

\section{Introduction}

1-bit compressive sensing (CS) is a fundamental high-dimensional statistical problem and is of interest both from a theoretical point of view and in practical applications, where sensing is performed subject to quantization.  Since its introduction in \cite{Bou08}, the theory of 1-bit CS has seen several advances, with the required number of measurements varying depending on the recovery guarantee, assumptions on the underlying signal being estimated, and the possible presence of noise.

Under the noiseless model, the measurements $\yv \in \{-1,1\}^m$ are generated according to
\begin{equation}
    b = {\rm sign}(\Av\xv^*),
\end{equation}
where $\xv^* \in \RR^n$ is the unknown underlying signal, $\Av \in \RR^{m \times n}$ is the measurement matrix, and the sign function ${\rm sign}(\cdot)$ is applied element-wise.  The value of ${\rm sign}(0)$ is not important for the results of this paper, so we adopt the convention ${\rm sign}(0) = 1$.

We assume that the signal $\xv^*$ lies in some class $\Xc_k$,. The most common choice for $\Xc_k$ is the set of all $k$-sparse signals for some $k \ll n$. Sparsity is satisfied by many natural families of signals (e.g., image and audio in appropriate representations), and assuming it often leads to statistical and computational benefits (see \cite{Has15}).  As we discuss further below, we will assume that $\Xc_k$ is a {\em subset} of the $k$-sparse signals, but possibly a strict subset.  Given such a class, the problem of 1-bit CS for signals in $\Xc_k$ comes with several important distinctions:
\begin{itemize}
    \item For {\em universal 1-bit CS} (also known as the {\em for-all guarantee}), it is required that the measurement matrix $\Av$ achieves some success criterion uniformly for all $\xv^* \in \Xc_k$.
    \item For {\em non-universal 1-bit CS} (also known as the {\em for-each guarantee}), the measurement matrix $\Av$ may be randomized, and the success criterion is only required to hold with high probability over this randomness for each fixed $\xv^* \in \Xc_k$.
    \item Under the {\em support recovery} success criterion, the goal is to produce an estimate $\Shat \subseteq [n]$ such that $\Shat = {\rm supp}(\xv^*)$, where ${\rm supp}(\cdot)$ denotes the support.
    \item Under the {\em approximate recovery} success criterion, the goal is to produce an estimate $\hat{\xv}$ such that\footnote{Note that the output of 1-bit CS measurements is invariant to scaling, so the magnitude cannot be recovered.} $\big\| \frac{\xv^*}{\|\xv^*\|} - \frac{\xvhat}{\|\xvhat\|} \big\|$ is no larger than some pre-specified threshold $\epsilon$.
\end{itemize}
The distinction between the universal and non-universal criteria has a drastic impact on the required number of measurements.  For instance, when $\Xc_k$ is the set of all $k$-sparse vectors, support recovery is possible under the for-each guarantee using $O(k \log n)$ measurements \cite{Ati12,Ber17}, whereas attaining the for-all guarantee requires $\Omega\big( k^2 \frac{\log n}{\log k} \big)$ measurements \cite{Ach17}.

On the other hand, there is recent evidence that for more {\em restricted} classes $\Xc_k$, this gap can be narrowed.  In particular, it was shown in \cite{Flo19} that for the set of {\em binary} $k$-sparse vectors (i.e., the non-zero entries are all equal to one), then universal exact recovery is possible using $O(k^{3/2} \log n)$ measurements.  This can also be deduced from an approximate recovery result of \cite{Jac13}, as we will see in Section \ref{subsec:upperbound}.  A simple counting argument shows that $\Omega\big(k \log \frac{n}{k}\big)$ measurements are needed  even for binary input vectors, and this is the best lower bound to date.

The preceding observations raise the following two important questions:
\begin{enumerate}
    \item Is the $k^{3/2}$ dependence unavoidable for universal recovery in the binary setting?
    \item Under what broader classes $\Xc_k$, can we similarly avoid the $k^2$ dependence?
\end{enumerate}
In this paper, we partially address the first question by showing that the $k^{3/2}$ dependence is unavoidable, at least when considering i.i.d.~Gaussian measurements, as was considered in \cite{Jac13}.\footnote{In contrast, \cite{Flo19} used a ``two-step'' design consisting of a list-disjunct part followed by a Gaussian part.  We do not prove $k^{3/2}$ to be unavoidable for such designs, though we expect this to be the case.} In more detail, we show that if significantly fewer measurements are used, then with constant probability, the matrix will fail to achieve the universal support recovery guarantee.  Analogous hardness results {\em with respect to a given random design} have frequently appeared in other compressive sensing works, e.g., \cite{Wai09,Ati12,Sca15}. 

To address the second question above, we consider the class of $k$-sparse vectors with {\em bounded dynamic range}:
\begin{equation}
    \Xc_k(R) = \bigg\{ \xv \in \RR^n \,:\, \| \xv \|_0 \le k \text{ and } \frac{\max_{i \in \supp(\xv)} |x_i| }{ \min_{i \in \supp(\xv)} |x_i|  } \le R \bigg\},  \label{eq:set_Xk}
\end{equation}
where $R \ge 1$, and $R$ may grow as a function of $n$.  The assumption of bounded dynamic range is motivated by the fact that the hardness result (showing $k^2$ dependence) for universal recovery in \cite{Ach17} implicitly considers signals with unbounded dynamic range.  On the other hand, in several practical applications, a bounded dynamic range assumption is reasonable, e.g., if the non-zero entries of $\xv$ correspond to transmitting users in a communication setup, then a bounded dynamic range may correspond to each of them having comparable transmit power.  Additionally, this notion has been studied previously in the context of for-each 1-bit CS in \cite{GNR10}, and recently, universal 1-bit CS with approximate support recovery \cite{Maz21}. 

Observe that setting $R = 1$ in \eqref{eq:set_Xk} recovers the binary case $\xv \in \{0,1\}^n$ (the constant of $1$ is without loss of generality, since the 1-bit model is invariant to scaling).  For signals in $\Xc_k(R)$, we provide a simple generalization of the above-mentioned $O\big( k^{3/2}\log n \big)$ achievability result as a corollary of \cite{Jac13}, and we provide an impossibility result for i.i.d.~Gaussian designs.

\subsection{Technical Overview}\label{sec:overview}

In order to recover the support of $x \in \Xc_k(R)$ via 1-bit CS using measurement matrix $A$, it is clearly necessary that for any $y \in \Xc_k(R)$ such that $\supp(x)\neq \supp(y)$, we also have $\textrm{sign}(Ax) \neq \textrm{sign}(Ay)$.  Accordingly, we provide the following definition.
\begin{definition}
The matrix $A \in \mathbb{R}^{m \times n}$ is a {\em valid universal $1$-bit measurement matrix} for support recovery on $\Xc_k(R)$ if and only if 
\begin{align*}
&\forall u \in \Xc_k(R),\forall v \in \Xc_k(R), \\
&\supp(u) \neq \supp(v) \implies {\rm sign}(Au) \neq {\rm sign}(Av)
\end{align*}
\end{definition}
Our main result is the following.\footnote{Asymptotic notation such as $\tilde{O}(\cdot)$ and $\tilde{\Omega}(\cdot)$ hides logarithmic factors of the argument.}
\begin{theorem}[Informal; see \cref{maintheorem}]\label{mainthminf}
Suppose $k$ and $R$ are sufficiently small in terms of $n$. Let $A \in \bbR^{m \times n}$ be a random matrix with i.i.d.~$N(0,1)$ entries. If it holds with constant probability that $A$ is a valid universal $1$-bit measurement matrix for support recovery on $\Xc_k(R)$, then $m = \tilde{\Omega}\big(Rk^{3/2} + k \log\frac{n}{k}\big)$.
\end{theorem}

The scaling regime of primary interest here is $k\ll n^{2/3}$, since if $k \gg n^{2/3}$ then one would improve on the $k^{3/2}$ scaling by simply letting $A$ be the $n \times n$ identity matrix (i.e., $m = n$).  Similarly, the interesting range of $R$ is $R \ll k^{0.5}$, since an upper bound of $O(k^2 \log(n))$ was shown in \cite{Ach17} for the universal support recovery of arbitrary $k$-sparse signals.  The assumed scaling of $k$ and $R$ are restricted accordingly in the formal statement, \cref{maintheorem}.

We now turn to an overview of the proof of \cref{mainthminf}. The lower bound of $\log {n \choose k} = \Omega(k \log(n/k))$ is simply due to the fact that the number of distinct measurement outcomes must be at least ${n \choose k}$.  Hence, the main goal is to show that with $m \le \tilde{O}(Rk^{3/2})$, $A$ is not a valid universal $1$-bit measurement matrix, with constant probability. Let $v_1, \dots, v_m \in \bbR^{n-2}$ denote the first $n-2$ coordinates of each of the $m$ rows of $A$.

We will first show that for $A$ to be a valid measurement matrix with constant probability, it must be the case that $\{v_1, \dots, v_m\}$ is $(n-2, k-1, \sqrt{4\log m}/R)$-{\em balanced}, where we say that a collection $V \subseteq \bbR^{n}$ is $(n, k, d)$-balanced if for every set $S\subseteq [n]$ of size $k$, there exists $v \in V$ such that $\left| \sum_{i \in S} v_i\right| \leq d$. The idea is as follows: Suppose that there exists a set $S \subseteq [n-2]$ such that for every $v \in \{v_1, \dots, v_m\}$, $\left|\sum_{i \in S}v_i\right| > \frac{\sqrt{4\log m}}{R}$. In this case, the vectors $x$ and $y$ with
 \[
    x_{i} = \begin{cases}
        R & \text{if $i \in S$}\\
        1 & \text{if $i = n - 1$}\\
        0 & \text{otherwise}\\
    \end{cases} \text{ and }
    y_{i} = \begin{cases}
        R & \text{if $i \in S$}\\
        1 & \text{if $i = n$}\\
        0 & \text{otherwise}\\
    \end{cases}
\]
cannot be distinguished by the measurement matrix $A$, since (as we will show) the entries in the last two columns of $A$ are all at most $\sqrt{4 \log m}$ with high probability. Thus, it suffices to prove a lower bound on the number of Gaussian vectors needed to form an $(n-2, k-1, \sqrt{4\log m}/R)$-{balanced} family.
 
 The study of bounds on the size of $(n, k, d)$-balanced families dates back to the early work of \cite{ABCO88}, which examined the setting $V \subseteq \{\pm 1\}^n$ and $k =n$. Indeed, one can view an $(n, k, d)$-balanced family as a collection of hyperplanes passing through the origin such that every $k$-sparse vector is within a margin $d$ of one of the hyperplanes. There is a long line of work on such problems and variants thereof, but {\em without} the sparsity constraint.
 For instance, \cite[Theorem 14]{calkin2008cuts} shows that if the number of hyperplanes is $o(n^{3/2})$, then the expected number of edges that are not cut is unbounded. See also \cite{Cos09,YY21} and the references therein. 
 
Returning to our setting, for a subset $S$ of size $k$, let $F_S$ denote the event that for every $v \in V$, $|\sum_{i \in S} v_i| > d$. We want to argue that if $m = |V|$ is too small, then $\Pr[\cup_S F_S]$ is lower bounded by a constant. To this end, we invoke {\em de Caen's inequality} that lower bounds the probability of a union.  

\begin{lemma}[de Caen's inequality, \cite{Dec97}]\label{decaen} 
Let $\{A_i\}_{i \in I}$ be a finite family of events in a probability space. Then:
\[
\Pr\left[\cup_{i \in I} A_i \right] \geq \sum_{i \in I} \frac{\Pr[A_i]^2}{\sum_{j \in I} \Pr[A_i \cap A_j]}.
\]
\end{lemma}

In our application, there are ${n \choose k}$ events of interest. For any $S$, $\Pr[F_S]$ is easy to lower bound using standard anti-concentration results for Gaussians. What occupies the bulk of our analysis is an upper bound for the denominator, $\sum_{T} \Pr[F_S \cap F_T]$. Somewhat surprisingly, the calculations here turn out to be quite delicate. To see why, consider the following naive approach: Let $E_S$ denote the event that for a single random Gaussian vector $v \in \bbR^n$, $|\sum_{i \in S} v_i| \leq r$, and let $\gamma = \Pr[E_S]$. Hence, the numerator is $(1-\gamma)^{2m}$.  For the denominator, the inclusion-exclusion formula gives:
\[
\Pr[F_S \cap F_T] = (1 - \Pr[F_S] - \Pr[F_T] + \Pr[F_S \cap F_T])^m
\]
Now, for low values of $k$, a typical $T$ has small intersection with $S$, and hence, $F_S$ and $F_T$ are approximately independent, i.e., $\Pr[F_S \cap F_T] \approx \gamma^2 \ll \gamma$. Thus, for most $T$, $\Pr[F_S \cap F_T] \approx (1-\Omega(\gamma))^m$, and thus, one might expect by invoking \cref{decaen} that $\Pr[\cup_S F_S] \geq (1-O(\gamma))^m$. Unfortunately, this is not enough, as $\gamma \approx d/\sqrt{k}$, and hence, we do not get constant failure probability for $m =\Theta(k^{3/2}/d)$ as we desired.

We work more carefully to get a stronger bound. First of all, we do a more refined analysis of the denominator based on the size of the overlap $|T \cap S|$. While for most $T$, the overlap is small and $\Pr[E_S \cap E_T] \approx \gamma^2$, it also holds that for the few sets $T$ with large overlap, $\Pr[E_S \cap E_T]$ approaches $\gamma$. Secondly, we keep close track of the {\em constant} factors in the terms of the denominator, because roughly speaking, we need to show that the numerator and denominator match in the first-order terms and only differ in the second-order terms. A more detailed overview of the proof is available in \cref{sec:balanced}. Our proof also generalizes from Gaussian to Rademacher measurements in the case that $R = 1$; this extension is presented in \cref{sec:rademacher}.

\subsection{Other Related Work}

Under the setting of {\em universal 1-bit CS}, which is our focus, various works have established both upper and lower bounds for both {\em support recovery} and {\em approximate recovery} (i.e., accuracy to within some target Euclidean distance $\epsilon$, up to scaling).

In \cite{plan2012robust}, it was shown that both $O\left(\frac{k}{\epsilon^6}\log{\frac{n}{k}}\right)$ and $O\left(\frac{k}{\epsilon^5}(\log{\frac{n}{k}})^2\right)$ measurements are sufficient for universal approximate recovery. Then, a significant improvement with respect to $\epsilon$ was achieved in \cite{gopi2013one}, showing that $O(k^3\log{\frac{n}{k}} + \frac{k}{\epsilon}\log{\frac{k}{\epsilon}})$ measurements are sufficient. Recently, \cite{Ach17} showed that using RUFFs (Robust Union Free Families), a modification of UFFs (first used in \cite{gopi2013one}), one could further reduce the number of measurements to $O(k^2\log{n} + \frac{k}{\epsilon}\log{\frac{k}{\epsilon}})$. In addition, they provided a lower bound of $\Omega(k\log{\frac{n}{k}} + \frac{k}{\epsilon})$.

As for the support recovery problem, \cite{gopi2013one} established that $O(k^2\log{n})$ measurements are sufficient for non-negative signals, which was generalised to arbitrary real-valued signals with a matching number of measurements in \cite{Ach17}. In addition, \cite{Ach17} provided a nearly tight lower bound of $\Omega(k^2\frac{\log{n}}{\log{k}})$ for this setting.

Signals with bounded dynamic range have received relatively less attention.  
Beyond binary vectors \cite{Flo19} (i.e., $R=1$), it has been shown that $O(R^2 k \log n)$ non-adaptive measurements suffice for the weaker {\em for-each} guarantee, while the dependence on $R$ can be avoided altogether using adaptive measurements \cite{Gup10}.  More recently, for universal 1-bit CS with {\em approximate support recovery}, it was shown in \cite{Maz21} that having a known bounded dynamic range $R$ reduces the leading term in the number of measurements from $k^{3/2}$ to $R k$.  We additionally show in \cref{subsec:upperbound} that an upper bound for {\em exact} support recovery under the for-all guarantee can be deduced from an approximate recovery result given in \cite{Jac13} (deducing the same from \cite{Maz21} is also straightforward).


Other related problems include group testing \cite{Du93,Ald19}, standard CS \cite{Wai09,Ati12,Sca15}, and multi-bit CS \cite{Sla15}, but the precise details are less relevant for comparing to our own results.

%% file: prelims.tex

\section{Preliminaries}

Throughout the paper, the function $\log(\cdot)$ has base $e$, and all information measures are in units of nats.

\subsection{Useful Technical Tools}

We will routinely use the following standard estimate for the Gaussian density.
\begin{lemma}[Small ball probabilities for Gaussians]
\label{smallballlemma}
If $X \sim N(0, \sigma^2)$, then for any $0<\delta<1$:
\[
\sqrt{\frac2\pi} \delta - \frac{\delta^3}{\sqrt{2\pi}} \leq \Pr[|X| \leq  \delta \sigma] \leq \sqrt{\frac2\pi} \delta.
\]
In addition, the upper bound also holds for $\max_t \Pr[t-\delta\sigma \leq X \leq t+\delta\sigma]$. 
\end{lemma}

We also recall the following well-known  asymptotic results about $\log{{n \choose k}}$:
\begin{itemize}
    \item If $k = o(n)$, then
        \begin{equation}
            \label{smallklognchoosek}
            \log{{n \choose k}} = (1 + o(1))k\log{\left(\frac{n}{k}\right)}
        \end{equation}
    
    \item If $k = \Theta(n)$, then
         \begin{equation}\log{{n \choose k}} = (1 + o(1))nH_2\left(\frac{k}{n}\right),
         \end{equation}
    where $H_2(p)$ is the binary entropy function in nats, i.e.,~$H_2(p) =  p\log{\big(\frac{1}{p}\big)} + (1 - p)\log{\big(\frac{1}{1 - p}\big)}$.
\end{itemize}

\subsection{Upper Bound for Bounded Dynamic Range Signals}
\label{subsec:upperbound}

Here we show that the universal approximate recovery guarantee of \cite{Jac13} translates directly to a universal support guarantee for signals in $\Xc_k(R)$.  Such a result may be considered folklore, but we are not aware of a formal statement in the literature.

The result of interest for our purposes is \cite[Thm.~2]{Jac13}, which reveals that in order to produce an estimate $\xvhat$ satisfying $\big\| \frac{\xv^*}{\|\xv^*\|} - \frac{\xvhat}{\|\xvhat\|} \big\| \le \epsilon_0$, an i.i.d.~Gaussian measurement matrix succeeds with probability at least $1-\delta$ (for any fixed constant $\delta > 0$) using a number of measurements satisfying
\begin{equation}
    m = O\bigg( \frac{k}{\epsilon_0} \log n + k \log \frac{1}{\epsilon_0} \bigg) \label{eq:m_Jacques}
\end{equation}
with a large enough implied constant.  To convert this to a support recovery guarantee, we need to establish how small $\epsilon_0$ should be so that the above guarantee also implies $\supp(\xvhat) = \supp(\xv^*)$.  To do so, let $x$ and $y$ be two $k$-sparse signals with different supports, and assume without loss of generality that $\|x\| = \|y\| = 1$.  Then there must exist some index $i$ where one signal (say $x$) is non-zero but the other is zero.  From this entry alone, we have $\|x - y\| \ge x_i$.  The assumption of dynamic range $R$ additionally implies that $x_i \ge \frac{1}{\sqrt{(k-1)R^2 + 1}}$, which is attained in the extreme case that $x$ has $k-1$ other non-zero values that are $R$ times higher than $x_i$.  It follows that $\|x - y\| \ge \frac{1}{\sqrt{(k-1)R^2 + 1}}$, which means that $\epsilon_0 = \Theta\big( \frac{1}{R\sqrt{k}} \big)$ suffices for support recovery.  Hence, \eqref{eq:m_Jacques} becomes
\begin{equation}
    m = O\big(  R k^{3/2} \log n \big). \label{eq:m_Jacques2}
\end{equation}
This finding naturally generalizes the sufficiency of $O(k^{3/2} \log n)$ measurements for binary signals, and demonstrates that the $k^2$ barrier can be avoided when $R \ll \sqrt{k}$. 

%% file: gaussian.tex

\section{Lower Bound for Bounded Dynamic Range Signals}

\cref{maintheorem} below is the more detailed version of our main result, informally summarized earlier in \cref{mainthminf}.

\begin{theorem}
\label{maintheorem}
Fix $\epsilon \in \big(0,\frac{2}{3}\big)$, and suppose that $k \le n^{\frac{2}{3} - \epsilon}$ and $R \le \min{(n^{0.5\epsilon}, k^{0.5 - \epsilon})}$. If $A =(a_{i, j}) \in \mathbb{R}^{m \times n}$, where each $a_{i, j}$ is an i.i.d.~standard Gaussian random variable, and
\begin{equation}
    m <cR\left(\frac{k}{\log k}\right)^{3/2}\min{\left({R^2}, \log^2{k}\right)} \label{eq:m_bound}
\end{equation}
for a sufficiently small constant $c>0$,
then with probability at least $1/3$, $A$ is not a valid universal $1$-bit measurement matrix for support recovery on $\Xc_k(R)$.
\end{theorem}

We observe that this lower bound matches the upper bound in \eqref{eq:m_Jacques} up to logarithmic factors (note that when ignoring logarithmic factors, we can safely replace $\min\left({R^2}, \log^2{k}\right)$ by $1$).  Regarding the assumed scaling on $k$ and $R$ here, we note the following:
\begin{itemize}
    \item The assumption $k \le n^{\frac{2}{3} - \epsilon}$ covers the sparsity regimes of primary interest, since $k \ge n^{\frac{2}{3} + \epsilon}$ would amount to the right-hand side of \eqref{eq:m_bound} exceeding $n$.  However, letting $m = n$ with $A$ being the identity matrix would then require fewer measurements.  Thus, in such scaling regimes, this trivial measurement scheme is more appropriate than the i.i.d.~Gaussian scheme.
    \item The assumption $R \le k^{0.5-\epsilon}$ covers the scaling regimes of primary interest on the number of measurements, because $R \ge k^{0.5+\epsilon}$ would imply that the right-hand side of \eqref{eq:m_bound} exceeds $\Omega(k^2)$.  However, it is known that $k^2$ dependence can be achieved even for {\em arbitrary} $k$-sparse signals regardless of the dynamic range \cite{Ach17}.
    \item On the other hand, the assumption $R \le n^{0.5\epsilon}$ imposes a non-trivial restriction.  Ideally this would be improved to $R \le n^{1.5\epsilon}$, which is a natural threshold because any higher would make the right-hand side of \eqref{eq:m_bound} again exceed $n$ when $k = \Theta(n^{2/3-\epsilon})$.  We believe that that our constant $0.5$ could be increased without too much extra effort, but that increasing it all the way to $1.5$ may be more challenging.
\end{itemize}
The ``failure'' probability of $1/3$ in \cref{maintheorem} is arbitrary, and can be adjusted to any fixed constant in $(0,1)$ while only affecting the unspecified constant $c$.  This is because any constant positive ``success'' probability could trivially be improved, and made arbitrarily close to one, by independently generating the measurement matrix $O(1)$ times.

The remainder of the section is devoted to proving \cref{maintheorem}.

\subsection{Reduction From Balancing}

As discussed in \cref{sec:overview}, we obtain the lower bound through a connection to a vector balancing problem, formally defined as follows.

\begin{definition}
    A set $V \subseteq \mathbb{R}^n$ is said to be {\em $(n, k, d)$-balanced} if for any $S \subseteq [n]$ of size $k$, there exists $v \in V$ satisfying $|\sum_{i \in S} v_i| \leq d$.
\end{definition}


We show that the existence of a universal 1-bit measurement matrix for $\Xc_k(R)$ of a certain size implies a bound on the size of a balanced family of vectors. A lower bound on the latter will thus imply a lower bound on the number of measurements.


\begin{lemma}
\label{successreducelemma}
If $A = (a_{i, j}) \in \mathbb{R}^{m \times n}$ is a valid universal $1$-bit measurement matrix for support recovery on $\Xc_k(R)$, then either 
\begin{itemize}
    \item $V = \{(a_{i,1}, a_{i,2}, \dots, a_{i,n - 2}) \mid i \in [m]\}$ is $\big(n - 2, k - 1, \frac{\sqrt{4\log{m}}}{R}\big)$-balanced, or
    \item there exist indices $i \in [1, m], j \in [n - 1, n]$ such that $|a_{i, j}| > \sqrt{4\log{m}}$.
\end{itemize}
\end{lemma}
\begin{proof}
Suppose that $|a_{i, j}| \leq \sqrt{4\log{m}}$ for all $i \in [1, m], j \in [n - 1, n]$. Consider any subset $S \subseteq [n - 2]$, where $|S| = k - 1$. Let $x, y$ be $R$-bounded $k$-sparse vectors such that 
    $$x_{j} = \begin{cases}
        R & \text{if $j \in S$}\\
        1 & \text{if $j = n - 1$}\\
        0 & \text{otherwise}\\
    \end{cases} \text{ and }
    y_{j} = \begin{cases}
        R & \text{if $j \in S$}\\
        1 & \text{if $j = n$}\\
        0 & \text{otherwise}.\\
    \end{cases}
    $$
    Since $A$ is a valid universal $1$-bit measurement matrix for support recovery on $\Xc_k(R)$, we must have ${\rm sign}(Ax) \neq {\rm sign}(Ay)$, so there must exist a row $a_i$ in $A$ such that  ${\rm sign}(a_i \cdot x) \neq {\rm sign}(a_i \cdot y)$. Without loss of generality, suppose that $a_i \cdot x \geq 0$ and $a_i \cdot y < 0$. Hence,
    \begin{align}
    \label{boundentries}
    \begin{split}
    \sum_{j \in S} Ra_{i,j} + a_{i,n - 1} \geq 0, \qquad \text{and} \qquad \sum_{j \in S} Ra_{i,j} +  a_{i, n} < 0.
    \end{split}
    \end{align}
    Since $|a_{i, n - 1}|, |a_{i, n}| \leq \sqrt{4\log{m}}$, it follows from \eqref{boundentries} that
    \begin{align*}
        \left|R\sum_{j \in S} a_{i, j}\right| \leq \sqrt{4\log{m}}
    \end{align*}
    Hence, $V$ is $\big(n - 2, k - 1, \frac{\sqrt{4\log{m}}}{R}\big)$-balanced, as claimed.
\end{proof}

\ignore{
\begin{corollary}
\label{failreducelemma}
Suppose that each $a_{i, j}$ is an i.i.d.~standard Gaussian random variable. If
\begin{itemize}
    \item $V = \{v | v = (a_{i, 1}, a_{i, 2} \dots, a_{i, n - 2}), i \in [1, m]\}$ is not $\left(n - 2, k - 1, \frac{\sqrt{4\log{m}}}{R}\right)$-balanced AND
    \item $\forall i \in [n - 1, n], \forall j \in [1, m], |a_{i, j}| \leq \sqrt{4\log{m}}$
\end{itemize} then $A$ is not a valid universal $1$-bit measurement matrix for support recovery on $\Xc_k(R)$.
\end{corollary}

\begin{proof}
This is the contrapositive of the \cref{successreducelemma}
\end{proof}
}

\subsection{Proof of \cref{maintheorem}}

In this section, we prove \cref{maintheorem} using the following lower bound on the size of balanced vector families, which will in turn be proved in \cref{sec:balanced}. 
\begin{theorem}
\label{sidetheorem}
Fix $\epsilon \in \big(0,\frac{2}{3}\big)$, and suppose that $k \le n^{\frac{2}{3} - \epsilon}$. Let $V \subseteq \mathbb{R}^n$ be a set of $m$ vectors, each independently drawn with i.i.d.~$N(0,1)$ entries. Then, for any $d$ satisfying
\begin{equation}
    d \ge \max{\left(\frac{4^{1/(1 + \epsilon)}k^{1.5}}{n^{1/(1 + \epsilon)}}, k^{\epsilon - 0.5}\sqrt{4\log{m}}\right)}, \quad d \le \sqrt{4\log{m}}, \label{eq:choice_d}
\end{equation}
if $m$ satisfies
\[
 m \le c\frac{k^{1.5}}{d}\min{\left(\frac{1}{d^2}, \log{k}\right)}
 \]
for a sufficiently small constant $c>0$, then with probability at least $1/2$ the set $V$ is not $\big(n, k, d)$-balanced.
\end{theorem}

For the sake of generality, we have considered the full range of $d$ in \eqref{eq:choice_d} that our proof permits, but in view of \cref{successreducelemma}, proving \cref{maintheorem} only requires handling the much more specific choice of $d = \frac{\sqrt{4 \log m}}{R}$, as we now show.

\begin{proof}[Proof of \cref{maintheorem}]
\noindent Let $H$ be the event that the matrix $A = (a_{i, j}) \in \mathbb{R}^{m \times n}$ sampled from i.i.d.~standard Gaussian random variables is not a valid universal $1$-bit measurement matrix for support recovery on $\Xc_k(R)$. Let $G$ be the event that $\forall i \in [1, m], \forall j \in [n - 1, n], |a_{i, j}| \leq \sqrt{4\log{m}}$, and let $P$ be the event that the set $V$ in \cref{successreducelemma} is not $\big(n - 2, k - 1, \frac{\sqrt{4\log{m}}}{R}\big)$-balanced.

From \cref{successreducelemma}, $P \cap G \implies H$, and therefore, $\Pr[H] \geq \Pr[P] - \Pr[\overline{G}]$. 
We show that $\Pr[\overline{G}]\leq \frac16$:

\begin{align*}
    \Pr[\overline{G}] &= \Pr\left[ \bigcup_{i = 1}^{m} \bigcup_{j = n - 1}^{n} \big\{ a_{i, j} > \sqrt{4\log{m}} \big\}\right]\\
    &\leq 2m \Pr[Z > \sqrt{4\log{m}}]\\
    &\leq 2m e^{-2 \log m} \le \frac16,
 \end{align*}
where in the last inequality we assume that $k$ (and hence $m$) exceeds a sufficiently large constant.
Therefore, $\Pr[H] \geq \frac{1}{2} - \frac{1}{6}= \frac{1}{3}$. The minor difference in $n$ and $k$ between Theorems \ref{maintheorem} and \ref{sidetheorem} can be absorbed by adjusting the constant $c$.

\cref{sidetheorem} provides the following condition on $m$:
    $$m \le c\frac{k^{1.5}}{d}\min{\left(\frac{1}{d^2}, \log{k}\right)}$$ 
for $d$ satisfying \eqref{eq:choice_d}.  We claim that the choice $d = \frac{\sqrt{4\log{m}}}{R}$ indeed satisfies \eqref{eq:choice_d}; this is seen as follows:
\begin{itemize}
    \item \cref{lemmad} below establishes that
    $\big(\frac{\sqrt{4\log{m}}}{R}\big)^{1 + \epsilon} \ge \frac{4k^{1.5(1 + \epsilon)}}{n}$, which implies $d \ge \frac{4^{1/(1 + \epsilon)}k^{1.5}}{n^{1/(1 + \epsilon)}}$.
    \item The inequality $\frac{\sqrt{4\log{m}}}{R} \geq k^{\epsilon - 0.5}\sqrt{4\log{m}}$ follows immediately from $R \leq k^{0.5 - \epsilon}$.
    \item $\frac{\sqrt{4\log{m}}}{R} \leq \sqrt{4\log{m}}$ follows immediately from $R \ge 1$.
\end{itemize}
Substituting this choice of $d$ into \cref{sidetheorem}, we obtain the condition
\begin{equation}
    m \le c'\frac{R}{\sqrt{\log{m}}}k^{1.5}\min{\left(\frac{R^2}{\log{m}}, \log{k}\right)} \label{eq:m_init}
\end{equation}
for a sufficiently small constant $c'$.  Manipulating this expression so as to obtain $m$ only on the the left hand side, we obtain  the condition
$$m \le c''\frac{R}{\sqrt{\log{k}}}k^{1.5}\min{\left(\frac{R^2}{\log{k}}, \log{k}\right)}.$$
This manipulation uses the assumption $R \le k^{0.5 - \epsilon}$, which implies that $\log m$ and $\log k$ are of the same order when $m$ equals the largest value satisfying \eqref{eq:m_init}.
\end{proof}

\begin{lemma}
\label{lemmad}
When $R \le n^{0.5\epsilon}$ and $k \le n^{\frac{2}{3}-\epsilon}$, we have $\big(\frac{\sqrt{4\log{m}}}{R}\big)^{1 + \epsilon} \ge \frac{4k^{1.5(1 + \epsilon)}}{n}$.
\end{lemma}
\begin{proof}
We have
\begin{align*}
    R &\le n^{0.5\epsilon}\\
    &\le \frac{\sqrt{4\log{m}}}{4} \cdot n^{0.5\epsilon} \text{ [Assuming $\log{m} \ge 4$]}\\
    &\le \frac{\sqrt{4\log{m}}}{4^{\frac{1}{1 + \epsilon}}} \cdot n^{\frac{0.5\epsilon + 0.5\epsilon^2}{1 + \epsilon}}\\
    &\le \sqrt{4\log{m}} \cdot \left(\frac{n^{0.5\epsilon + 1.5\epsilon^2}}{4}\right)^{\frac{1}{1 + \epsilon}}  \text{ [Since $0.5\epsilon^2 \le 1.5\epsilon^2$]}\\
    &= \sqrt{4\log{m}} \cdot \left(\frac{n}{4n^{(\frac{2}{3} - \epsilon)(1.5 + 1.5\epsilon)}}\right)^{\frac{1}{1 + \epsilon}} \\
    &\le \sqrt{4\log{m}} \cdot \left(\frac{n}{4k^{1.5(1 + \epsilon)}}\right)^{\frac{1}{1 + \epsilon}}, \text{ [Since $k < n^{\frac{2}{3} - \epsilon}$]}
\end{align*}
and re-arranging gives
\begin{align*}
    \left(\frac{\sqrt{4\log{m}}}{R}\right)^{1 + \epsilon} &\ge \frac{4k^{1.5(1 + \epsilon)}}{n}.
\end{align*}
\end{proof}

%% file: balanced.tex

\section{Proof of \cref{sidetheorem} (Lower Bound for Balancing)}\label{sec:balanced}
Let $v^{(1)}, \dots, v^{(m)}$ denote the vectors in $V$, each independently drawn with i.i.d.~$N(0,1)$ entries.  Let $x_s \in \{0,1\}^n$ be the binary vector with support $s$, and define $X = \{x_s\}_{|s|=k}$. 
Denote by $F_s$ the failure event for the subset $s$ in the balancing problem, i.e.,~the event that $|x_s \cdot v^{(i)}| > d$ for all $i=1,\dotsc,m$.

We establish a lower bound on $\Pr[\bigcup_{s \in X}F_s]$ using de Caen's lower bound (\cref{decaen}):
$$\Pr\left[\bigcup_{s \in X}F_s\right] \geq \sum_{s \in X} \frac{\Pr[F_s]^2}{\sum_{t \in X}\Pr[F_s \cap F_t]},$$
with the goal being to lower-bound the right-hand side by $\frac{1}{2}$.

Let $E_{s, i}$ be a single failure event, i.e.,~$|x_s \cdot v^{(i)}| > d$, such that $F_s = \bigcap_{i = 1}^{m}E_{s, i}$.  We split the terms $\Pr[F_s \cap F_t]$ according to the amount of overlap between the two relevant subsets, and accordingly write $\Pr_{\beta}[F_s \cap F_t] = \Pr[F_s \cap F_t]$ in the case that $|x_s \cap x_t| = \beta k$ for some $\beta \in [0, 1]$.

\begin{definition}
We introduce the quantities
\begin{equation}
    \alpha = {n \choose k} \Pr[F_s]^2 \label{eq:def_alpha}
\end{equation}
and
\begin{equation}
    \nu_{\beta} = {k \choose (1 - \beta)k}{n - k \choose (1 - \beta)k}\Pr_{\beta}[F_s \cap F_t]. \label{eq:def_v}
\end{equation}
\end{definition}

Here for a fixed $x_s$, $\nu_{\beta}$ represents the sum of the contributions of all the pairs $(x_s, x_t)$ where $|x_s \cap x_t| = \beta k$ in the denominator of de Caen's bound (i.e.,~$\sum_{t \in X}\Pr[F_s \cap F_t]$). Formally, $\nu_{\beta} = \sum_{t: |x_t \cap x_s| =\beta k}\Pr_{\beta}[F_s \cap F_t]$.  Note that $\nu_{\beta}$ is well-defined only when $\beta k \in \mathbb{Z}_{\geq 0}$, which is understood to hold throughout our analysis.

For a fixed $x_s$, the number of $x_t$'s in $X$ such that $|x_s \cap x_t| = \beta k$ is ${k \choose (1 - \beta)k}{n - k \choose (1 - \beta)k} = {k \choose \beta k}{n - k \choose (1 - \beta)k}$, because ${k \choose \beta k}$ is the number of ways to choose $\beta k$ entries within the non-zero entries of $x_s$ to be a part of $x_t$ and ${n - k \choose (1 - \beta)k}$ is the number of ways to pick the remaining entries from outside of $x_s$, ensuring that the intersection size is exactly $\beta k$.

To simplify subsequent notation, we define
\[ \widetilde{k} = \frac{k}{d^2} \implies \widetilde{k}^{-0.5} = k^{-0.5}d.\]
We define three summations, which partition the sum in the denominator of de Caen's bound (i.e.,~$\sum_{t \in X}\Pr[F_s \cap F_t]$) as follows:
$$A = \sum_{\beta \in [0, \widetilde{k}^{-0.5}]}\nu_\beta,~~ B = \sum_{\beta \in [\widetilde{k}^{-0.5}, 0.9]}\nu_\beta, \text{ and } C = \sum_{\beta \in [0.9, 1]}\nu_\beta.$$
Here and subsequently, $\displaystyle \sum_{\beta \in [a, b]}\nu_{\beta}$ is compact notation meant to represent the sum over only the well-defined $\beta \in [a, b]$. Formally, 
$$\sum_{\beta \in [a, b]}\nu_{\beta} = \sum_{t = \ceil{ak}}^{\floor{bk}}\nu_{\frac{t}{k}} = \sum_{t = \ceil{ak}}^{\floor{bk}}{k \choose k - t}{n - t \choose k - t}\Pr_{\beta = \frac{t}{k}}[F_s \cap F_t].$$

\noindent \textbf{Intuition behind why we split into three sums}: 
\begin{itemize}
    \item $A$ is the sum of all the pairs who intersection size is ``sufficiently small'' with respect to $k$, i.e.,~at most $\sqrt{\widetilde{k}}$. Since the intersection is so small, so we are able to utilise the fact that $\Pr_{\beta}[F_s \cap F_t] \approx \Pr[F_s]^2$ to get $A \leq 1.5\alpha$ (\cref{lemmaA}). 
    \item For $B$, the sum of the binomial terms, i.e.,~${k \choose (1 - \beta)k}{n - k \choose (1 - \beta)k}$, is much smaller than ${n \choose k}$, and the largest asymptotic term in which $Pr_{\beta}[F_s \cap F_t]$ differs from $\Pr[F_s]^2$ is $O(\beta\frac{d^2}{k})$.  Both of these balance out and eventually give $B \leq 0.25\alpha$ (\cref{lemmaB}). 
    \item Lastly, $C$ is used for the specific case of $\beta = \Omega(1)$, where we loosely bound $\Pr_{\beta}[F_s \cap F_t] \leq 1$ and show that the product of the two binomial terms is small, so $C \leq 0.25\alpha$ (\cref{lemmaC}).
\end{itemize}

We now proceed more formally.  For any fixed $s \in X$, we have $\sum_{t \in X}\Pr_{\beta}[F_s \cap F_t] \leq A + B + C$, and hence
\begin{align*}
    \frac{1}{\sum_{t \in X}\Pr_{\beta}[F_s \cap F_t]} &\geq \frac{1}{A + B + C}.
\end{align*}
Let $Q$ be the event that the set $V \subseteq \mathbb{R}^n$, where $|V| = m$ and each $V_{i, j}$ is an i.i.d.~standard Gaussian random variable, is not $\left(n, k, d\right)$-balanced.
Then, since $\bigcup_{s \in X}F_{s} \implies Q$, we have
$$ \Pr[Q] \geq \Pr\left[\bigcup_{s \in X}F_{s}\right]. $$
Then, using de Caen's bound (\cref{decaen}), we have
\begin{align*}
    \Pr\left[\bigcup_{s \in X}F_s\right] &\geq \sum_{s \in X} \frac{\Pr[F_s]^2}{\sum_{t \in X}\Pr[F_s \cap F_t]} \\
    &\geq \frac{|X|\Pr[F_1]^2}{A + B + C}\\
    &= \frac{{n \choose k} \Pr[F_1]^2}{A + B + C}.
\end{align*}
The numerator is equal to $\alpha$ by definition, and $A$, $B$, and $C$ will be bounded in the subsequent subsections (\cref{lemmaA}, \cref{lemmaB} and \cref{lemmaC} below) to deduce that
\begin{align*}
    \Pr\left[\bigcup_{s \in X}F_s\right] &\geq \frac{\alpha}{(1.5 + 0.25 + 0.25)\alpha} = \frac{1}{2}.
\end{align*}
Combining the above, we have $\Pr[Q] \geq \frac{1}{2}$ as desired.

With some small constants $c'$ and $c''$, we will require $m \le c' \frac{k^{1.5}}{d^3}$ in \cref{lemmaA}, and $m \le c'' \frac{k^{1.5}}{d}\log{k}$ in \cref{lemmaB} and \cref{lemmaC}.  Since we need these to hold simultaneously, the overall scaling on $m$ is
$$m \le c\frac{k^{1.5}}{d}\min{\left(\frac{1}{d^2}, \log{k}\right)}$$ for a sufficiently small constant $c$.

We now proceed with the analysis of $A$, $B$, and $C$.

\subsection{Initial Characterization of $E_{s,i}$ and $F_s$}
Recall that we set $ \widetilde{k} = \frac{k}{d^2}$, and that $d$ satisfies \eqref{eq:choice_d}.
 Additionally, define $$r = \frac{1}{\sqrt{2\pi}}\frac{2d}{\sqrt{k}} = \sqrt{\frac{2}{\pi \widetilde{k}}}. $$ 
We proceed with several straightforward auxiliary lemmas; in the following, $\overline{(\cdot)}$ denotes the complement of an event, and $x_s \in \{0,1\}^n$ denotes the binary vector with support $s$.

\begin{lemma}
\label{lemmaE}
Under the preceding definitions, we have $r - O(\frac{1}{\widetilde{k}^{1.5}}) \leq \Pr[\overline{E_{s, i}}] \leq r$.
\end{lemma}
\begin{proof}
Observe that $x_s \cdot v^{(i)}$ is distributed as $N(0,k)$, since it is the sum of $k$ i.i.d.~$N(0,1)$ variables. Hence, for $Z \sim N(0,k)$, we have $\Pr[\overline{E_{s,i}}] = \Pr\left[|Z| \leq d\right]$. Invoking \cref{smallballlemma} with $\delta = \frac{d}{\sqrt{k}}$ and $\sigma^2 = k$, we obtain
\begin{align*}
\sqrt{\frac{2}{\pi}} \cdot \frac{d}{\sqrt{k}} - \frac{\left(\frac{d}{\sqrt{k}}\right)^3}{\sqrt{2\pi}} &\leq \Pr[\overline{E_{s,i}}] \leq \sqrt{\frac{2}{\pi}} \cdot \frac{d}{\sqrt{k}}.
\end{align*}
The lemma follows from our definitions of $\widetilde{k}$ and $r$.
\end{proof}

\begin{lemma}
\label{lemmaF2}
We have $\Pr[F_s]^2 \geq (1 - 2r + r^2)^m$.
\end{lemma}

\begin{proof}
Since $F_s = \bigcap_{i = 1}^{m}E_{s, i}$, we have
\begin{align*}
    \Pr[F_s]^2 &= (\Pi_{i = 1}^{m}\Pr[E_{s, i}])^2\\
    &\geq (1 - r)^{2m} \text{ [Using \cref{lemmaE}]}\\
    &= (1 - 2r + r^2)^m.
\end{align*}
\end{proof}

\begin{lemma}
\label{lemmaalpha}
The quantity $\alpha = {n \choose k} \Pr[F_s]^2$ satisfies $\log{\alpha} \geq \log{n \choose k} - 2mr - O\big(\frac{m}{\widetilde{k}}\big)$.
\end{lemma}
\begin{proof}
Using \cref{lemmaF2} and the inequality $- x - x^2 \leq \log{(1 - x)}$ for $x \in [0, \frac{1}{2})$, we obtain
\begin{align*}
\log{\alpha} \geq \log{{n \choose k}} + m\log{(1 - 2r + r^2)} \geq \log{{n \choose k}} - 2mr - 2mr^2,
\end{align*}
and the lemma follows since $r^2 = \Theta(1/\widetilde{k})$.
\end{proof}

\subsection{Analysis of $A$}
We first obtain a bound on $\Pr_\beta[F_s \cap F_t]$ that applies for the analysis of both $A$ and $B$.

\begin{lemma}
\label{lemmaFbeta}
Given $\beta \leq 0.9$, $\Pr_{\beta}[F_s \cap F_t] \leq \big(1 - 2r + r^2 + O(\beta r^2) + O(\widetilde{k}^{-3/2})\big)^m$.
\end{lemma}
\begin{proof}
Since $F_s \cap F_t = \bigcap_{p = 1}^{m}(E_{s, p} \cap E_{t, p})$, we have $\Pr[F_s \cap F_t] = \prod_{p = 1}^{m}\Pr[E_{s, p} \cap E_{t, p}] = \prod_{p = 1}^{m} (1 - \Pr[\overline{E_{s, p}}] - \Pr[\overline{E_{t, p}}] + \Pr[\overline{E_{s, p}} \cap \overline{E_{t, p}}])$. \cref{lemmaE} already bounds $\Pr[\overline{E_{s,p}}]$. 

The following lemma analyzes $\Pr[\overline{E_{s, p}} \cap \overline{E_{t, p}}]$. The event $\overline{E_{s, p}} \cap \overline{E_{t, p}}$ means that the binary vectors $x_s$ and $x_t$ are balanced with respect to $v^{(p)}$. To upper bound this term, we utilise the fact that for both of them to be balanced, if we fix the contribution of the intersection region (i.e.,~of $x_s \cap x_t$ with respect to $v^{(p)}$) then the contribution of the exclusive regions (i.e.,~of $x_s \setminus x_t$ and $x_t \setminus x_s$ with respect to $v^{(p)}$) will be constrained to be in a region of size at most $2 d$.

\begin{lemma}
\label{lemmaEbeta}
Given $|x_s \cap x_t| = \beta k$, where $\beta \leq 0.9$, it holds for each $p \in [m]$ that $\Pr[\overline{E_{s, p}} \cap \overline{E_{t, p}}] \leq r^2(1 + \beta + O(\beta^2))$.
\end{lemma}
\begin{proof}
Let $f_{\beta k}(z)$ denote the density function of $N(0, \beta k)$ at point $z$.  We have
\begin{align*}
    \Pr[\overline{E_{s, p}} \cap \overline{E_{t, p}}] &= \Pr\left[|x_s \cdot v^{(p)}| \leq d, |x_t \cdot v^{(p)}| \leq d\right]\\
    &= \int_{-\infty}^{\infty}f_{\beta k}(i)\Pr\left[-z - d \leq N(0, k - \beta k) \leq -z + d\right]^2 \,dz.\\
\end{align*}
Defining $q = \max_{z \in \mathbb{R}}\Pr\big[-z - d \leq N(0, k - \beta k) \leq -z + d\big]$, we have from \cref{smallballlemma} that $q \leq 2\frac{d}{\sqrt{(k - \beta k)2\pi}}=\frac{r}{\sqrt{1-\beta}}$.
Thus, for $\beta \leq 0.9$, a Taylor expansion gives
\[
    \Pr[\overline{E_{s, p}} \cap \overline{E_{t, p}}] 
    \leq \int_{-\infty}^{\infty}f_x(z)q^2 \,dz
    \leq q^2 
    \leq \frac{r^2}{1 - \beta}= r^2(1 + \beta + O(\beta^2)).
    \]
\end{proof}
We can now conclude the proof of \cref{lemmaFbeta}:
\begin{align*}
\Pr[F_s \cap F_t] &= \Pi_{p = 1}^{m}\Pr[E_{s, p} \cap E_{t, p}]\\
&= \Pi_{p = 1}^{m} (1 - \Pr[\overline{E_{s, p}}] - \Pr[\overline{E_{t, p}}] + \Pr[\overline{E_{s, p}} \cap \overline{E_{t, p}}]) \\
&\leq \left(1 - 2r + O\left(\frac{1}{\widetilde{k}^{1.5}}\right) + r^2(1 + \beta + O(\beta^2))\right)^m  \text{ [Using \cref{lemmaE} and \cref{lemmaEbeta}]}\\
&\leq \left(1 - 2r + r^2 + O(\beta r^2) + O\left(\frac{1}{\widetilde{k}^{1.5}}\right)\right)^m.
\end{align*}
\end{proof}

We are now ready to prove our desired upper bound on $A$.
\begin{lemma}
\label{lemmaA}
Given $m \le c \frac{k^{1.5}}{d^3}$ for some sufficiently small constant $c$, and $k \le n^{\frac{2}{3} - \epsilon}$, it holds that $A \leq 1.5\alpha$.
\end{lemma}
\begin{proof}
Defining $l = (1 - \beta)k$, we have
\begin{align}\label{eqn:A}
    A &\leq \sum_{\beta \in [0, \widetilde{k}^{-0.5}]}{k \choose (1 - \beta)k}{n - k \choose (1 - \beta)k}\Pr_{\beta}[F_s \cap F_t] \leq \sum_{l = k - \widetilde{k}^{0.5}}^{k}{k \choose l}{n - k \choose l}\Pr_{\beta}[F_s \cap F_t],
\end{align}
Since $\beta \leq \widetilde{k}^{-0.5}$, we have $O(\beta r^2) = O(\frac{\beta}{\widetilde{k}}) = O(\frac{1}{\widetilde{k}^{1.5}})$.  Hence, \cref{lemmaF2} and \cref{lemmaFbeta} respectively yield:
\begin{align*}
    \Pr[F_s]^2 &\geq \left(1 - 2r + r^2\right)^m, \\
    \Pr_{\beta}[F_s \cap F_t] &\leq \left(1 - 2r + r^2 + O\left(\frac{1}{\widetilde{k}^{1.5}}\right)\right)^m,
\end{align*}
and combining these gives the following for some constant $C$:
\begin{align*}
    \frac{\Pr_{\beta}[F_s \cap F_t]}{\Pr[F_s]^2} &\leq \left(1 + \frac{O(\widetilde{k}^{-1.5})}{1 - 2r + r^2}\right)^m \leq \left(1 + C\widetilde{k}^{-1.5}\right)^m \leq e^{-mC\widetilde{k}^{-1.5}}. 
\end{align*}
Taking $m \le \frac{1}{3C}\widetilde{k}^{1.5}$, it follows that $\Pr_\beta[F_s \cap F_t] \leq e^{\frac{1}{3}}\Pr[F_s]^2 \leq 1.5\Pr[F_s]^2$. Substituting into \eqref{eqn:A} and observing that
\begin{align*}
    \sum_{l = k - \widetilde{k}^{0.5}}^{k} {k \choose l}{n - k \choose l} \leq \sum_{l = 0}^{k} {k \choose l}{n - k \choose l} = {n \choose k},
\end{align*}
we obtain $A \leq {n \choose k}1.5\Pr[F_s]^2= 1.5\alpha$. 
\end{proof}

\subsection{Analysis of $B$}

The main step in the analysis of $B$ is showing that $\nu_{\beta} \leq \frac{0.25\alpha}{k}$ for all $\beta \in [\widetilde{k}^{-0.5}, 0.9]$. It will be convenient to study the logarithm (base $e$) of $\nu_\beta$.


\begin{lemma}
\label{lemmavbeta}
Given $\beta \leq 0.9$, for large enough $k$, we have $\log{\nu_{\beta}} \leq (1 + \epsilon)kH_2(\beta) + \log{n - k \choose (1 - \beta)k} - 2mr + O\big(\frac{m}{\widetilde{k}}\big)$.
\end{lemma}
\begin{proof}
We have
\begin{align*}
    \nu_{\beta} &= {k \choose (1 - \beta)k}{n - k \choose (1 - \beta)k}\Pr_{\beta}[F_s \cap F_t] \\
    &\leq {k \choose (1 - \beta)k}{n - k \choose (1 - \beta)k}\left(1 - 2r + r^2 + O(\beta r^2) + O\Big(\frac{1}{\widetilde{k}^{1.5}}\Big)\right)^m. \text{ [Using \cref{lemmaFbeta}]}
\end{align*}
Recalling the asymptotic results on binomial coefficients as specified in \cref{smallklognchoosek}, we have the following (once $k$ is large enough so that $1 + o(1) \leq 1 + \epsilon$):
\begin{align*}
    \log{\nu_{\beta}} &\leq (1 + \epsilon)kH_2(\beta) + \log{n - k \choose (1 - \beta)k} + m\log{\left(1 - 2r + r^2 + O(\beta r^2) + O\Big(\frac{1}{\widetilde{k}^{1.5}}\Big)\right)}\\
    &\leq (1 + \epsilon)kH_2(\beta) + \log{n - k \choose (1 - \beta)k} - 2mr + mr^2 + O(m\beta r^2) + O\left(\frac{m}{\widetilde{k}^{1.5}}\right) \\
        &\hspace*{7.5cm}\text{ [Using $\log{(1 - x)} \leq -x$]} \\
    &\leq (1 + \epsilon)kH_2(\beta) + \log{n - k \choose (1 - \beta)k} - 2mr + O\left(\frac{m}{\widetilde{k}}\right). \text{ [Using $r = O\big( \frac{1}{\sqrt{\widetilde{k}}} \big)$]}
\end{align*}
\end{proof}

We also need the following estimate of a difference of binomial coefficients.
\begin{lemma}
\label{lemmalogsub}
$\log{n \choose k} - \log{n - k \choose (1 - \beta)k} \geq \beta k \log{\left(\frac{n - k}{k}\right)}$.
\end{lemma}
\begin{proof}
    We have
\begin{align*}
    \frac{{n \choose k}}{{n - k \choose (1 - \beta)k}} 
    &= \prod_{a = 0}^{k - \beta k - 1}\left(\frac{n - a}{n - k - a}\right) \times \prod_{b = 0}^{\beta k - 1}\left(\frac{n - k + \beta k - b}{k - b}\right)\\
    &= \prod_{a = 0}^{k - \beta k - 1}\left(1 + \frac{k}{n - k - a}\right) \times \prod_{b = 0}^{\beta k - 1}\left(\frac{n - k + \beta k - b}{k - b}\right)\\
    &\geq \left(1 + \frac{k}{n}\right)^{k - \beta k} \times \left(\frac{n - k}{k}\right)^{\beta k}.
\end{align*}
Taking the log, it follows that
\begin{align*}
    \log{n \choose k} - \log{n - k \choose (1 - \beta)k} &\geq (k - \beta k)\log{\left(1 + \frac{k}{n}\right)} + \beta k \log{\left(\frac{n - k}{k}\right)} 
    \geq \beta k \log{\left(\frac{n - k}{k}\right)}.
\end{align*}
\end{proof}

We can now obtain the key result of this subsection, namely, a bound on $B$.
\begin{lemma}
\label{lemmaB}
Given $m \le c \frac{k^{1.5}}{d}\log{k}$ for a small enough constant $c$, $k \le n^{\frac{2}{3} - \epsilon}$, and $d$ satisfying \eqref{eq:choice_d},
it holds that $B \leq 0.25\alpha$.
\end{lemma}

\begin{proof}
We introduce the function
\begin{equation}
    f(\beta) = \beta k \log{\Big(\frac{n - k}{k}\Big)} - (1 + \epsilon)kH_{2}(\beta) \label{eq:f_def}
\end{equation}
which yields 
\begin{equation}
    \frac{\partial f}{\partial \beta} (\beta) = k\log{\bigg(\frac{n - k}{k} {\Big(\frac{\beta}{1 - \beta}\Big)}^{1 + \epsilon}\bigg)}. \label{eq:f_deriv}
\end{equation}
For $\beta \in [\widetilde{k}^{-0.5}, 0.9]$, we have from \cref{lemmaalpha} and \cref{lemmavbeta} that
\begin{align*}
    \log{\alpha} - \log{\nu_\beta} &\geq \left(\log{n \choose k} - \log{n - k \choose (1 - \beta)k}\right) - (1 + \epsilon)kH_2(\beta) - O\left(\frac{m}{\widetilde{k}}\right) \\
    &\geq \beta k \log{\left(\frac{n - k}{k}\right)} - (1 + \epsilon)kH_2(\beta) - O\left(\frac{m}{\widetilde{k}}\right) \text{ [Using \cref{lemmalogsub}]}\\
    &= f(\beta) - O\left(\frac{m}{\widetilde{k}}\right).
\end{align*}
In \cref{fincreasing} below, we show that $f(\beta)$ is an increasing function in the range $[\widetilde{k}^{-0.5}, 0.9]$ with the minimum at $f(\widetilde{k}^{-0.5})$. Moreover, recalling $\widetilde{k} = \frac{k}{d^2}$, we have
\begin{align*}
    kH_2(\widetilde{k}^{-0.5}) &= k\left(k^{-0.5}d\log{\left(\frac{1}{k^{-0.5}d}\right)} + \left(1 - k^{-0.5}d\right)\log{\left(\frac{1}{1 - k^{-0.5}d}\right)}\right)\\
    &= k^{0.5}d\left(\log{\left(\frac{1}{k^{-0.5}d}\right)} - \log{\left(\frac{1}{1 - k^{-0.5}d}\right)}\right) + k\log{\left(\frac{1}{1 - k^{-0.5}d}\right)}\\
    &\leq k^{0.5}d\left(\log{\left(\frac{1}{k^{-0.5}d}\right)}\right) - k\log{\left(1 - k^{-0.5}d\right)}\\
    &\leq k^{0.5}d(0.5\log{k} - \log{d}) - k\left(-k^{-0.5}d - k^{-1}d^2\right) \\
        & \hspace{2cm} \text{ [Using $\log{(1 - x)} \geq - x - x^2$, when $x \in [0, \frac{1}{2}$]]}\\
    &= k^{0.5}d(1 + 0.5\log{k} - \log{d}) + d^2\\
    &\leq k^{0.5}d\log{\left(\frac{2k^{0.5}}{d}\right)} + d^2.
\end{align*}
Hence,
\begin{align*}
    f(\widetilde{k}^{-0.5}) &= \widetilde{k}^{-0.5}k\log{\left(\frac{n - k}{k}\right)} - (1 + \epsilon)kH_2(\widetilde{k}^{-0.5}) \\
    &\geq k^{0.5}d\log{\left(\frac{n - k}{k}\right)} - (1 + \epsilon)\left(k^{0.5}d\log{\left(\frac{2k^{0.5}}{d}\right)} + d^2\right)\\
    &\geq k^{0.5}d\log{\left(\frac{n}{2k}\right)} - k^{0.5}d\log{\left(\left(\frac{2k^{0.5}}{d}\right)^{1 + \epsilon}\right)} - 2d^2 \text{ [Since $\epsilon < 1$]}\\
    &= k^{0.5}d\log{\left(\frac{nd^{1 + \epsilon}}{2^{2 + \epsilon}k^{1.5 + 0.5\epsilon}}\right)} - 2d^2 \\
    &\geq k^{0.5}d\log{\left(\frac{nd^{1 + \epsilon}}{4k^{1.5 + 0.5\epsilon}}\right)} - 2d^2. \text{ [Since $\epsilon > 0$]}\\
\end{align*}

Using the condition $d^{1 + \epsilon} > \frac{4k^{1.5 + 0.5\epsilon + \epsilon}}{n}$ in \eqref{eq:choice_d}, we additionally have
\begin{gather*}
    \frac{nd^{1 + \epsilon}}{4k^{1.5 + 0.5\epsilon}} > k^{\epsilon}\\
    \implies \log{\left(\frac{nd^{1 + \epsilon}}{4k^{1.5 + 0.5\epsilon}}\right)} > \epsilon\log{k}.
\end{gather*} 
Therefore,
\begin{align*}
    \log{\alpha} - \log{\nu_{\beta}} &\geq f(\widetilde{k}^{-0.5}) - O\left(\frac{m}{\widetilde{k}}\right)\\
    &\geq \epsilon k^{0.5}d\log{k} - 2d^2 - O\left(\frac{m}{\widetilde{k}}\right).
\end{align*}
Since $m < c\frac{k^{1.5}}{d}\log{k}$ for a small enough constant $c$, we can select $c$ such that $O\big(\frac{m}{\widetilde{k}}\big) \le \frac{\epsilon}{2} \frac{k^{1.5}}{d}\log{k} \times \frac{d^2}{k} = \frac{\epsilon}{2} k^{0.5} d \log{k}$
    
Therefore, the above expression simplifies to
\begin{align*}
    \log{\alpha} - \log{\nu_{\beta}} &\geq \epsilon k^{0.5}d\log{k} - 2d^2 - O\left(\frac{m}{\widetilde{k}}\right)\\
    &\geq \frac{\epsilon}{2}k^{0.5}d\log{k} - 2d^2
\end{align*}
We analyse the two terms individually as follows, again considering the conditions on $d$ in \eqref{eq:choice_d}:
\begin{itemize}
    \item For the first term,
    \begin{align*}
        \frac{\epsilon}{2} k^{0.5}d\log{k} 
        &= \frac{\epsilon}{2}k^{\epsilon}\sqrt{4\log{m}}\log{k} \text{ [Since $d \geq k^{\epsilon - 0.5}\sqrt{4\log{m}}$]}\\
        &\geq \Omega(k^{\epsilon}).
    \end{align*}
    
    \item For the second term, using $d \le \sqrt{4\log{m}}$, we obtain
    \begin{align*}
       2d^2 &\le 8\log{m}\\
       &= O(\log{k}). \text{ [Since $d \geq k^{\epsilon - 0.5}\sqrt{4\log{m}}$ and $m = O\big(\frac{k^{1.5}}{d}\log{k}\big)$]}
    \end{align*} 
\end{itemize}
Combining these, we have
$$\frac{\epsilon}{2} k^{0.5}d\log{k} - 2d^2 = \Omega(k^{\epsilon}) \geq \log{k} + \log{4},$$
where the last step holds for $k$ exceeding a sufficiently large constant.  It follows that, for all $\beta \in [\widetilde{k}^{-0.5}, 0.9]$, we have $\log{\alpha} - \log{\nu_{\beta}} \geq \log{k} + \log{4}$, or equivalently $\nu_{\beta} \leq 0.25\frac{\alpha}{k}$.  Summing over the relevant $\beta$ values then gives
\begin{align*}
    B = \sum_{\beta \in [\widetilde{k}^{-0.5}, 0.9]}\nu_{\beta} \leq k \times 0.25\frac{\alpha}{k}
    = 0.25\alpha.
\end{align*}
\end{proof}

In the above analysis, we focused on $\widetilde{k}^{-0.5}$ and then used monotonicity to handle all $\beta \in [\widetilde{k}^{-0.5}, 0.9]$.  The following lemma gives the required monotonicity property.


\begin{lemma}
\label{fincreasing}
Given $d$ satisfying \eqref{eq:choice_d} and $k \le n^{\frac{2}{3} - \epsilon}$, we have for $\beta \ge \widetilde{k}^{-0.5}$ that $\frac{\partial f}{\partial \beta}(\beta) > 0$.
\end{lemma}
\begin{proof}
Since $\widetilde{k} = \frac{k}{d^2}$, the assumption $\beta \ge \widetilde{k}^{-0.5}$ becomes $\beta \ge k^{-0.5}d$.  Additionally using the trivial bound $\frac{\beta}{1-\beta} \ge \beta$, it follows that
$$ \left(\frac{\beta}{1 - \beta}\right)^{1 + \epsilon} \ge k^{-0.5(1 + \epsilon)}d^{1 + \epsilon}. $$
Moreover, by assumption in \eqref{eq:choice_d}, we have
$$ d^{1 + \epsilon} \ge \frac{4k^{1.5 + 0.5\epsilon + \epsilon}}{n} \ge \frac{2k^{1.5 + 0.5\epsilon + \epsilon}}{n}, $$
and combining the two preceding equations gives
\begin{align*}
    {\left(\frac{\beta}{1 - \beta}\right)}^{1 + \epsilon} &\ge k^{-0.5(1 + \epsilon)}\frac{2k^{1.5 + 0.5\epsilon + \epsilon}}{n}\\
        &\ge \frac{2k^{1 + \epsilon}}{n}\\
        &\ge \frac{2k}{n}\\
        &\ge \frac{k}{n - k}.
\end{align*}
It follows that
\begin{equation*}
    \frac{n - k}{k} {\left(\frac{\beta}{1 - \beta}\right)}^{1 + \epsilon} \ge 1,
\end{equation*}
or equivalently,
\begin{equation*}
    k\log{\left(\frac{n - k}{k} {\left(\frac{\beta}{1 - \beta}\right)}^{1 + \epsilon}\right)} \ge 0.
\end{equation*}
Recalling from \eqref{eq:f_deriv} that the left-hand side is $\frac{\partial f}{\partial \beta}(\beta)$, this completes the proof.
\end{proof}

\subsection{Analysis of $C$}
\begin{lemma}
\label{lemmaC}
Given $m < c\frac{k^{1.5}}{d}\log{k}$ for a small enough constant $c$, $k \le n^{\frac{2}{3} - \epsilon}$, and $d$ satisfying \eqref{eq:choice_d},
it holds that $C \leq 0.25\alpha$.
\end{lemma}

\begin{proof}
We take $c$ to be small enough so that $\frac{2m}{\sqrt{\widetilde{k}}} + O\big(\frac{m}{\widetilde{k}}\big) = 2m\frac{d}{\sqrt{k}} + O\big(\frac{md^2}{k}\big) \le \frac{k}{2}\log{\frac{n}{k}}$.  Then, we have
\begin{align*}
    \log{\alpha} &\geq \log{{n \choose k}} - 2mr - O\left(\frac{m}{\widetilde{k}}\right) \text{ [Using \cref{lemmaalpha}]}\\
    &\geq k\log{\frac{n}{k}} - 2\frac{m}{\sqrt{\widetilde{k}}} - O\left(\frac{m}{\widetilde{k}}\right) \text{ [Since $r = \sqrt{\frac{2}{\pi \widetilde{k}}}$ and $\frac{2}{\pi} \le 1$]}\\
    &\geq \frac{k}{2}\log{\frac{n}{k}} \\
    &\geq \frac{k}{6}\log{n}. \text{ [Since $k < n^{\frac{2}{3} - \epsilon}$]}
\end{align*}
We proceed by bounding $C$ as follows:
\begin{align*}
    C &= \sum_{\beta \in [0.9, 1]}\nu_{\beta}\\
    &= \sum_{\beta \in [0.9, 1]}{k \choose (1 - \beta)k}{n - k \choose (1 - \beta)k}\Pr_{\beta}[F_s \cap F_t]\\
    &\leq 0.2k{k \choose 0.1k} {n - k \choose 0.1k} \Pr_{\beta}[F_s \cap F_t] \text{ [Since $k \geq 20 \implies \left\lceil 0.1k \right\rceil \leq 0.2k$]}\\
    &\leq 0.2k {k \choose 0.1k} {n - k \choose 0.1k} \times 1.
\end{align*}    
Taking the logarithm, we obtain
\begin{align*}
    \log{C} &\leq \log{\left(0.2k{k \choose 0.1k}{n - k \choose 0.1k}\right)}\\
    &\leq \log{(0.2)} + \log{k} + (0.1k)(\log{(10)} + 1) + (0.1k)(\log{(n)} + 1)\\
        &\hspace{4cm} \text{  [Using $\log{{n \choose k}} \leq k\left(\log{\left(\frac{n}{k}\right)} + 1\right)$]}\\
    &\leq 0.1k\log{n} + 0.6k + \log{k} + \log{(0.2)}\\
    &\leq 0.15k\log{n} - \log{4} \text{ [For large $n$, $0.6k + \log{k} + \log{0.2} < 0.05k\log{n} - \log{4}$]}\\
    &\leq \frac{k}{6}\log{n} - \log{4}.
\end{align*}    
Combining the above, we obtain the desired bound $C \leq 0.25\alpha$ (recalling that the log has base $e$).
\end{proof}

%% file: rademacher.tex

\section{Lower Bound for Rademacher Measurements}\label{sec:rademacher}

In the Rademacher measurement scheme, each entry of the measurement matrix is an i.i.d.~Rademacher random variable, i.e., $+1$ with probability $0.5$ and $-1$ with probability $0.5$.  Our analysis for Gaussian measurements, with slight modifications, is able to provide a lower bound in the case of Rademacher measurements and binary signals.

Before proceeding, we note that support recovery for $\Xc_k(R)$ with $R \ge 2$ is impossible using Rademacher measurements.  We demonstrate this via the following example.

\begin{example}
Let $n = 3$, $k = 2$ and $R = 2$, and consider $x_1 = [2, 1, 0]^T, x_2 = [2, 0, 1]^T$, noting that $\supp(x_1) \neq \supp(x_2)$.
In this case, none of the $2^3 = 8$ combinations of $b = (b_1, b_2, b_3) \in \{\pm 1\}^3$ can differentiate ${\rm sign}(b \cdot x_1)$ from ${\rm sign}(b \cdot x_2)$, i.e.,~${\rm sign}(2b_1 + 1b_2 + 0b_3)$ is same as ${\rm sign}(2b_1 + 0b_2 + 1b_3)$ for all 8 possible choices of $b$. This is because once the $\pm 1$ coefficient to $2$ is specified, the $\pm 1$ coefficient to $1$ is too small to impact the sign.  We can easily extend this example to general values of $R > 1$ and $k \geq 2$.
\end{example}

In view of this example, we focus our attention on $R = 1$, with binary signals $x \in \{0,1\}^n$.  

\begin{theorem}
\label{rademacherthm}
Fix $\epsilon \in \big(0,\frac{2}{3}\big)$, and suppose that $k \le n^{\frac{2}{3} - \epsilon}$. If $A = (a_{i, j}) \in \{\pm 1\}^{m \times n}$ where each $a_{i, j}$ is an independent Rademacher random variable and where $m = O(k^{1.5})$ with a small enough implied constant, then with probability at least $\frac{1}{2}$, $A$ is not a valid $1$-bit measurement matrix for support recovery on $\Xc_k(1)$
\end{theorem}

\subsection{Proof Sketch}

For the proof of the Rademacher case, we reduce the $(n, k - 1, 1)$-balanced problem to $1$-bit CS.

\begin{lemma}
\label{rademacherreduction}
If $A \in \{\pm 1\}^{m \times n}$ is a valid $1$-bit measurement matrix for support recovery on $\Xc_k(1)$, then the row vectors of $A$ are $(n, k - 1, 1)$-balanced.
\end{lemma}

This lemma is proved in a near-identical manner to \cref{successreducelemma}, but the term $\frac{\sqrt{4 \log m}}{R}$ is replaced by one (which trivially equals the magnitude of any given entry of $A$).  Another slight difference here is that we do not need to isolate two specific entries of $x$; with binary-valued $x$ and $A$, {\em any} two entries can play the role that entries $n-1$ and $n$ played in \cref{successreducelemma}.  Since these differences are straightforward, we omit the details.

Again having established a suitable reduction, the main remaining task is to prove the following lower bound for the balancing problem.

\begin{theorem}
\label{rademachersidethm}
Fix $\epsilon \in \big(0,\frac{2}{3}\big)$, and suppose that $k$ is an even number satisfying $k \le n^{\frac{2}{3} - \epsilon}$.  For a set $V \subseteq \{\pm 1\}^n$ where $|V| = m$ and each $V_{i, j}$ is an i.i.d.~Rademacher random variable, if $m = O(k^{1.5})$ with a small enough implied constant, then with probability at least $\frac{1}{2}$, the set $V$ is not $(n, k, 1)$-balanced.
\end{theorem}

Before outlining the proof, we show how this result yields \cref{rademacherthm}.

\begin{proof}[Proof of \cref{rademacherthm}]
Let $H$ be the event that the matrix $A = (a_{i, j}) \in \{\pm 1\}^{m \times n}$ with i.i.d.~Rademacher entries is not a valid universal $1$-bit measurement matrix for support recovery on $\Xc_k(1)$.
Let $P$ be the event that the set $V = \{(a_{i, 1}, a_{i, 2}, \dots, a_{i, n})| i \in [1, m]\}$ is not $(n, k - 1, 1)$-balanced.

Assuming momentarily that $k-1$ is even, we have from \cref{rademacherreduction} that $P \implies H$.  From \cref{rademachersidethm}, we have $\Pr[P] \geq \frac{1}{2}$ when $m = O((k-1)^{1.5})$. Under the big O notation we are able to change $k - 1$ to $k$ with only a slight change in the constant, and the desired lower bound on $\Pr[H]$ follows from that on $\Pr[P]$.

If $k-1$ is odd, then we can simply use the fact that $\Xc_{k-1}(1) \subseteq \Xc_{k}(1)$, and apply \cref{rademacherthm} with parameter $k-2$ (which is even).  Approximating $k-2$ by $k$ under the big-O notation in the same way as above, we deduce the same result.
\end{proof}

\subsection{Key Differences in Proof of \cref{rademachersidethm} w.r.t.~Proof of \cref{sidetheorem}}

Recall that $E_{s,i}$ is the balancing failure event for set $s$ and a single measurement $i$, and $F_s = \bigcap_{i=1}^m E_{s,i}$ is the overall failure event for set $s$.

We again split $\sum_{t \in X}Pr[F_s \cap F_t]$ into three terms $A$, $B$ and $C$. However in this analysis our cut-off point between $A$ and $B$ is $k^{-0.5}$ (instead of the original $\widetilde{k}^{-0.5}$).  Thus, we define 
    $$A = \sum_{\beta \in [0, k^{-0.5}]}\nu_\beta, B = \sum_{\beta \in [k^{-0.5}, 0.9]}\nu_\beta \text{ and } C = \sum_{\beta \in [0.9, 1]}\nu_\beta,$$
with $\nu_{\beta}$ defined in \eqref{eq:def_v}. 
Accordingly, in the original proof, all occurrences of $\widetilde{k}$ simplify to $k$, and $r = \sqrt{\frac{2}{\pi \widetilde{k}}}$ is redefined to $r = \sqrt{\frac{2}{\pi k}}$.

Since this is a discrete setting, the proofs of $r - O(\frac{1}{k}) \leq \Pr[\overline{E_{s, i}}] \leq r$ and $\Pr[\overline{E_{s, p}} \cap \overline{E_{t, p}}] \leq r^2(1 + \beta + O(\beta^2))$ (when $|x_s \cap x_t| = \beta k$) are established slightly differently, outlined as follows.

\begin{lemma} 
\label{rademacherlemma2nchoosen} ${2n \choose n} = \frac{4^n}{\sqrt{\pi n}}\left(1 - \frac{1}{8n} + O(n^{-2})\right), \forall n \in \mathbb{N}$
\end{lemma}

\begin{proof}
This follows easily from Stirling's approximation:
\begin{align*}
{2n \choose n} &= \frac{2^{2n}}{\sqrt{\pi n}}\left(1 - \frac{1}{8n} + \frac{1}{128n^2} + \frac{5}{1024n^3} + ...\right) \\
&= \frac{4^n}{\sqrt{\pi n}}\left(1 - \frac{1}{8n} + O(n^{-2})\right).
\end{align*}
\end{proof}

\begin{lemma}
\label{rademacherlemmaE}
For even-valued $k$ we have, $r - O\left(\frac{1}{k^{1.5}}\right) \leq \Pr[\overline{E_{s, i}}] \leq r$.
\end{lemma}

\begin{proof}
Recalling that the elements of $V$ are $v^{(1)},\dotsc,v^{(m)}$ and $x_s \in \{0,1\}^n$ has support $s$, we have
\begin{align*}
\Pr[\overline{E_{s, i}]} &= \Pr[|x_s \cdot v^{(i)}| \leq 1]\\
&= \Pr[x_s \cdot v^{(i)} = 0] \text{ [Since $k$ is even]}\\
&= {k \choose {\frac{k}{2}}}2^{-k}\\
&= \sqrt{\frac{2}{\pi k}}\left(1 - \frac{1}{4k} + O(k^{-2})\right) \text{ [Using \cref{rademacherlemma2nchoosen}]}.
\end{align*}
It follows for large enough $k$ that $r - O\big(\frac{1}{k^{1.5}}\big) \leq \Pr[\overline{E_{s, i}}] \leq r$.
\end{proof}

\begin{lemma}
\label{rademacherlemmaEbeta}
For even-valued $k$ and $|x_s \cap x_t| = \beta k$, where $\beta \leq 0.9$, we have for each index $p \in [m]$ that $\Pr[\overline{E_{s, p}} \cap \overline{E_{t, p}}] \leq r^2(1 + \beta + O(\beta^2))$.
\end{lemma}

\begin{proof}
Let $|x_s \cap x_t| = \beta k = d$, and observe that
\begin{align*}
    \Pr[\overline{E_{s, p}} \cap \overline{E_{t, p}}] &= \Pr[|x_s \cdot v^{(p)}| \leq 1, |x_t \cdot v^{(p)}| \leq 1]\\
    &= \Pr[x_s \cdot v^{(p)} = 0, x_t \cdot v^{(p)} = 0] \text{ [Since k is even]}
\end{align*}
If we fix $i$ to be the number of ones in $x_s \cap x_t$. Then the number of ones in $x_s - x_t$ will need to be exactly $\frac{k}{2} - i$ so that the total number of ones in $x_s = i + \frac{k}{2} - i = \frac{k}{2}$. Similarly there need to be exactly $\frac{k}{2} - i$ ones in $x_t - x_s$.  Hence,
\begin{align*}
    \Pr[\overline{E_{s, p}} \cap \overline{E_{t, p}}] &= \sum_{i = 0}^{d}{d \choose i}2^{-d} \times \left({k - d \choose \frac{k}{2} - i}2^{-(k - d)}\right)^{2} \\
    &\leq \frac{2}{\pi (k - d)}\sum_{i = 0}^{d}{d \choose i}2^{-d} \text{ [Using ${k - d \choose \frac{k}{2} - i} \leq {k - d \choose \frac{k - d}{2}} \leq \sqrt{\frac{2}{\pi k - d}}2^{k - d}$]}\\
    &= \frac{2}{\pi (k - d)}\\
    &= \frac{2}{\pi k (1 - \beta)} \\
    &= r^2(1 + \beta + O(\beta^2)). \text{ [Using a Taylor Series]}
\end{align*}
\end{proof}
For the analysis to show $A \leq 1.5\alpha$ under $m = O(k^{1.5})$ (instead of the original $m = O(\widetilde{k}^{1.5})$), all the technical details have the same idea behind them and the proof goes through. 

In the analysis of $B$ and $C$, the original choice $m = O(\frac{R}{\sqrt{\log{m}}}k^{1.5}\log{k})$ has the factor $\frac{R}{\sqrt{4\log{m}}}$ because of the $\big(n - 2, k - 1, \frac{\sqrt{4\log{m}}}{R}\big)$-balanced problem reduction. 
Since we are now using a different reduction, i.e.~$(n - 1, k - 1, 1)$-balanced, the analysis also simplifies the $m$ to be $O(k^{1.5}\log{k})$. But $m = O(k^{1.5})$ is strictly smaller and assumed in the analysis of $A$, so here we work under the common setting of $m = O(k^{1.5})$ for all three terms.

As for the analysis of $B$, one of the few things that differs in the proof is that now we need to show that $\frac{\partial f(\beta)}{\partial \beta} \ge 0$ for all $\beta > k^{-0.5}$ (instead of the original $\forall \beta > \widetilde{k}^{-0.5}$). Because it is possible that $k^{-0.5} \le \widetilde{k}^{-0.5}$, the correctness of this statement is non-trivial, and is described below. The remaining proof details are essentially unchanged.

\begin{lemma}
Given $k \le n^{\frac{2}{3} - \epsilon}$, we have for $\beta \ge k^{-0.5}$ that $\frac{\partial f}{\partial \beta}(\beta) \ge 0$, where $f(\cdot)$ is defined in \eqref{eq:f_def}.
\end{lemma}
\begin{proof}
The condition $\beta \ge k^{-0.5}$ implies that
\begin{gather*}
    \frac{\beta}{1 - \beta} \ge \beta \ge k^{-0.5}\\
     \implies \left(\frac{\beta}{1 - \beta}\right)^{1 + \epsilon} \ge k^{-0.5(1 + \epsilon)}.
\end{gather*}
In addition, the assumption $k \le n^{\frac{2}{3} - \epsilon}$ implies that
\begin{align*}
    2k^{0.5(1 + \epsilon)} &\le 2n^{0.5(\frac{2}{3} - \epsilon)(1 + \epsilon)}\\
    &= 2n^{\frac{1}{3} - \frac{\epsilon}{6} + \frac{\epsilon^2}{2}}\\
    &\le n^{\frac{1}{3} + \epsilon},
\end{align*}
and hence,
\begin{align*}
    2k^{0.5(1 + \epsilon)} &\le \frac{n}{k} \text{ [Since $k \le n^{\frac{2}{3} - \epsilon}$]}\\
    \frac{2k}{n} &\le k^{-0.5(1 + \epsilon)}.
\end{align*}
Therefore,
\begin{align*}
    \left(\frac{\beta}{1 - \beta}\right)^{1 + \epsilon} &\ge \frac{2k}{n}\\
    &\ge \frac{k}{n - k}, 
\end{align*}
and re-arranging gives
\begin{gather*}
    \frac{n - k}{k} \left(\frac{\beta}{1 - \beta}\right)^{1 + \epsilon} \ge 1\\
    \implies k\log{\left(\frac{n - k}{k} \left(\frac{\beta}{1 - \beta}\right)^{1 + \epsilon}\right)} \ge 0.
\end{gather*}
From \eqref{eq:f_deriv}, the left-hand side equals $\frac{\partial f}{\partial \beta}(\beta)$, and the proof is complete.
\end{proof}
Lastly, for the analysis of $C$, the original proof goes through without any significant modification.